\documentclass[a4paper]{article}
\usepackage[utf8]{inputenc}
\usepackage{graphicx}
\usepackage{authblk}
\usepackage{tcolorbox}
\usepackage{amsfonts}
\usepackage[margin=3cm]{geometry}

\usepackage{mathtools}
\usepackage{xspace}
\usepackage{hyperref}
\usepackage{tcolorbox}
\usepackage{amsmath}
\usepackage{amsthm}
\usepackage{comment}
\usepackage{xcolor}

\usepackage{array}
\newcolumntype{L}[1]{>{\raggedright\let\newline\\\arraybackslash\hspace{0pt}}m{#1}}
\newcolumntype{C}[1]{>{\centering\let\newline\\\arraybackslash\hspace{0pt}}m{#1}}
\newcolumntype{R}[1]{>{\raggedleft\let\newline\\\arraybackslash\hspace{0pt}}m{#1}}

\newtheorem{theorem}{Theorem}
\newtheorem{lemma}{Lemma}
\newtheorem{claim}{Claim}

\newtheorem{definition}{Definition}
\newtheorem{corollary}{Corollary}

\usepackage{tcolorbox}
\usepackage{framed}
\usepackage{xspace}
\usepackage[]{quoting}

\newcommand{\I}{\mathcal{I}}
\newcommand{\oh}{\mathcal{O}}
\newcommand{\D}{\mathcal{D}}
\newcommand{\A}{\mathcal{A}}
\newcommand{\E}{\mathcal{E}}
\newcommand{\Q}{\mathcal{Q}}
\newcommand{\F}{\mathcal{F}}
\newcommand{\B}{\mathcal{B}}

\newcommand{\Pe}{\mathcal{P}}

\usepackage[T1]{fontenc}
\usepackage[T1]{fontenc}

\newcommand{\cgm}{\mathcal{CG}}
\newcommand{\repset}[1]{\subseteq_{rep}^{#1}}
\newcommand{\choosepq}{{{p+q}\choose{p}}}
\newcommand{\mytilde}[1]{\overline{#1}}

\newcommand{\npc}{\textsf{NP}-complete\xspace}
\newcommand{\nph}{\textsf{NP}-hard\xspace}
\newcommand{\wh}{\textsf{W[1]}-hard\xspace}
\newcommand{\cp}{{\sc Clique}\xspace}
\newcommand{\cmcs}{{\sc Connected Modulator to Connected Set Problem}\xspace}

\newcommand{\rr}{Reduction Rule}
\newcommand{\yes}{\textsf{YES}\xspace}
\newcommand{\no}{\textsf{NO}\xspace}
\newcommand{\fpt}{\textsf{FPT}\xspace}
\newcommand{\dg}{\textsf{deg}\xspace}
\newcommand{\true}{ {\sc true}\xspace}
\newcommand{\false}{ {\sc false}\xspace}
\newcommand{\represents}{{\sc represents}\xspace}
\newcommand{\bcpq}{{\textsc{BCP$_q$}}\xspace}
\newcommand{\bcptwo}{{\textsc{BCP$_2$}}\xspace}

\newcommand{\bcepq}{{\textsc{BCEP$_q$}}\xspace}
\newcommand{\bceptwo}{{\textsc{BCEP$_2$}}\xspace}

\usepackage{xcolor}

\newcommand{\vertexpart}{{\sc BCP$_2$}\xspace}
\newcommand{\edgepart}{{\sc BCEP$_2$}\xspace}
\newcommand{\rvertexpart}{{\sc Restricted BCP$_2$}\xspace}

\newcommand{\defparprob}[4]{
\begin{tcolorbox}[colback=gray!5!white,colframe=gray!75!black]
  \vspace{-1mm}
  \begin{tabular*}{\textwidth}{@{\extracolsep{\fill}}lr} #1  & {\bf{Parameter:}} #3 \\ \end{tabular*}
  {\bf{Input:}} #2  \\
  {\bf{Question:}} #4
  \vspace{-1mm}
\end{tcolorbox}
}

\newcommand{\defprob}[3]{
\begin{tcolorbox}[colback=gray!5!white,colframe=gray!75!black]
  \begin{tabular*}{\textwidth}{@{\extracolsep{\fill}}lr} #1   \\ \end{tabular*}
  {\bf{Input:}} #2  \\
  {\bf{Question:}} #3
  \end{tcolorbox}
}

\newtheorem{red_rule}{Reduction Rule}
\title{Parameterized Complexity of Graph Partitioning into Connected Clusters}%Fixed Parametrer Tractability on Graphs with High Connectivity\thanks{Supported by None}Connected Modulator to Connected Set Problem}
\author[1]{Ankit Abhinav}
\author[1]{Susobhan Bandopadhyay}
\author[1]{Aritra Banik}
\author[2]{Saket Saurabh}

	\affil[1]{National Institute of Science, Education and Research, An OCC of Homi Bhabha National Institute,
	Bhubaneswar 752050, Odisha, India. \\
	\{ankit.abhinav, susobhan.bandyopadhyay, aritra\}@niser.ac.in}

\affil[2]{The Institute of Mathematical Sciences, HBNI, Chennai, India \\
	saket@imsc.res.in}

\begin{document}
	
	\maketitle

\begin{abstract}
% !TEX root = connected.tex
	
	Given an undirected graph $G$ and $q$ integers $n_1,n_2,n_3, \cdots, n_q$, balanced connected $q$-partition problem ($BCP_q$) asks whether there exists a partition of the vertex set $V$ of $G$ into $q$ parts $V_1,V_2,V_3,\cdots, V_q$ such that for all $i\in[1,q]$, $|V_i|=n_i$ and the graph induced on $V_i$ is connected. A related problem denoted as the balanced connected $q$-edge partition problem ($BCEP_q$) is defined as follows. Given an undirected graph $G$ and $q$ integers $n_1,n_2,n_3, \cdots, n_q$,  $BCEP_q$  asks whether there exists a partition of the edge set of $G$ into $q$ parts $E_1,E_2,E_3,\cdots, E_q$ such that for all $i\in[1,q]$, $|E_i|=n_i$ and the graph induced on the edge set $E_i$ is connected. Here we study both the problems for $q=2$ and prove that $BCP_q$ for $q\geq 2$ is $W[1]$-hard. We also show that $BCP_2$ is  unlikely to have a polynomial kernel on the class of planar graphs. Coming to the positive results, we show that $BCP_2$ is fixed parameter tractable (FPT) parameterized by treewidth of the graph, which generalizes to FPT algorithm for planar graphs. We design another FPT algorithm and a polynomial kernel on the class of unit disk graphs parameterized by $\min(n_1,n_2)$. Finally, we prove that unlike $BCP_2$, $BCEP_2$ is FPT parameterized by $\min(n_1,n_2)$.   

\end{abstract}

% !TEX root = connected.tex
\section{Introduction}
The following combinatorial question was raised by Andras Frank in his doctoral dissertation~\cite{frank1975combinatorial}. 
\begin{quote}
	Given a $q$-connected graph $G$, $q$ vertices $\{v_1,v_2,\cdots ,v_q\}\subseteq V(G)$, and $q$ positive integers $n_1,n_2,\cdots, n_q$ such that $\sum_{i=1}^q n_i=|V(G)|$, does there exists a partition ${V_1,\cdots,V_q}$ of $V(G)$ such that $v_i\in V_i$, $|V_i|=n_i$ and $G[V_i]$ is connected for all $1\leq i\leq q$? 
\end{quote}

The question was partially answered by the author and the statement was proved for the case when $q=2$. The original conjecture was settled separately by  Gy\"ori~\cite{gyori1976division} and Lov\'asz~\cite{lovasz1977homology} and it is famously known as the Gy\"ori-Lov\'asz theorem. The proof given by Gy\"ori, was constructive. A polynomial time algorithm can be constructed to find a $q$ partition of a $q$-connected graph. Surprisingly, when the graph is not $q$-connected, the problem becomes \nph, for $k\geq 2$~\cite{camerini1983complexity}. The general problem is known as the Balanced Connected $q$-Partition Problem (\bcpq) and can be formally stated as follows.

\defprob{\textbf{Balanced Connected $q$-Partition Problem}(\bcpq)}{An undirected graph $G$, an integer $q$ and $q$ positive integers $n_1,n_2,\cdots, n_q$ such that $\sum_{i=1}^q n_i=|V(G)|$.}{Does there exist a partition ${V_1,\cdots,V_q}$ of $V(G)$ such that  $|V_i|=n_i$ and $G[V_i]$ is connected for all $1\leq i\leq q$?}

The approximation version of the problem is called the Max Balanced Connected $q$-Partition Problem. The objective of the problem is to maximize the minimum partition size.  The best known approximation algorithm known is due to Chleb{\'\i}kov{\'a}~\cite{chlebikova1996approximating} and it presents a $\frac{4}{3}$-factor algorithm. 
It is also known that \bcpq does not admit an $\alpha$-approximation algorithm with $\alpha < 6/5$, unless P = NP~\cite{chataigner2007approximation}. 

In this paper, we initiate the study of fixed parameter tractability of \bcpq. We have proved the following results.

%\vspace{10pt}
\noindent
{\bf Our Results.} We have shown that \bcpq, for $q\geq 2$ is \wh. In Section~\ref{Sec:treewidth} we have shown that \bcptwo is fixed parameter tractable when parameterized by the treewidth. More specifically we have constructed a $\oh^*(k^{\oh(k)})$ algorithm for \bcptwo where $k$ is the treewidth of the given graph. In Section~\ref{sec:planar}, we considered the case when the underlying graph is planar. We have presented a $\oh^*(k^{\oh(k)})$ algorithm for \bcptwo where $k=\min(n_1,n_2)$. In Section~\ref{sec:planarnopolykernel} we have proved that under standard complexity theoretic assumptions there is no poly-kernel for \bcptwo in planar graphs. In Section~\ref{sec:udg}, we have proved that when the underlying graph is unit disk graph, there is a poly-kernel for \bcptwo parameterized by  $k=\min(n_1,n_2)$. In Section~\ref{sec:edge}, we have introduced the "edge" variant of the \bcptwo which we is defined as follows.

%\defprob{\bf \vertexpart}{An undirected graph $G(V,E)$ with $n$ vertices and an integer $k$.}{Does there exist $X \subseteq V$ of cardinality \alert{(exactly/at least, at most?)} $k$ such that $G[X]$  and $G[V \setminus X]$ is connected?}
\defprob{\textbf{Balanced Connected $q$-Edge Partition Problem} (\bcepq)}{An undirected graph $G$, an integer $q$ and $q$ positive integers $n_1,n_2,\cdots, n_q$ such that $\sum_{i=1}^q n_i=|E(G)|$.}{Does there exist a partition ${E_1,\cdots,E_q}$ of $E(G)$ such that  $|E_i|=n_i$ and $G[E_i]$ is connected for all $1\leq i\leq q$?}

In section~\ref{sec:edge}, we have proved that unlike \bcptwo, \bceptwo is \fpt parameterized by $\min(n_1,n_2)$.

%
%\defprob{\bf \edgepart}{An undirected graph $G(V,E)$ with $n$ vertices and an integer $k$.}{Does there exist $Y \subseteq E$ of cardinality \alert{(exactly/at least, at most?)} $k$ such that $G[Y]$  and $G[V \setminus Y]$ is connected?}

%Note: Here we have presented the case when in the problem of consideration, \vertexpart (\edgepart) we are looking for a subset of vertices (edges) of cardinality $k$ which induces a connected subset. Please note that all the results (\alert{check}) in the paper also holds when we are looking for a set of vertices (edges) of cardinality $k$ which induces a connected subgraph of a special kind, for example, a path of length $k$. 

%\vspace{10pt}
%\noindent
% !TEX root = connected.tex
\subsection{Notations}

\noindent
{\bf Graph Notations.} 
Let $G$ be a graph. We use $V(G)$ and $E(G)$ to denote the set of vertices and edge of $G$, respectively. 
Throughout the paper we use $n$ and $m$ to denote $|V(G)|$ and  $|E(G)|$, respectively.  For a vertex $v$, we use $N(v)$ to denote set of its neighbors, and use $\dg(v)$ to denote $|N(v)|$. For any graph $G$ and a set of vertices $V'\subseteq V(G)$, we denote the subgraph of $G$ induced by $V'$ by $G[V']$. For any set of vertices $X\in V(G)$,  we denote $V(G)\setminus X$ by $\tilde{X}$.
Most of the symbols and notations used of graph theory are standard and taken from Diestel~\cite{diestel2012graph}. 

\noindent
{\bf Parameterized Complexity and Algorithms.} 
The goal of parameterized complexity is to find ways of solving \nph problems more efficiently than brute force: here the aim is to
restrict the combinatorial explosion to a parameter that is hopefully
much smaller than the input size. Formally, a {\em parameterization}
of a problem is assigning a positive integer parameter $k$ to each input instance and we
say that a parameterized problem is {\em fixed-parameter tractable
	(\fpt)} if there is an algorithm that solves the problem in time
$f(k)\cdot \vert I \vert ^{\oh(1)}$, where $\vert I \vert$ is the size of the input and $f$ is an
arbitrary computable function that depends only on the parameter $k$. Such an algorithm is called an \fpt algorithm and such a running time is called \fpt running time. There is also an accompanying theory of parameterized intractability using which one can identify parameterized problems that are unlikely to admit \fpt algorithms. These are essentially proved by showing that the problem is W-hard.

Another major research field in parameterized complexity is kernelization.
Kernelization is a technique that formalizes the notion of preprocessing. A kernelization algorithm for a parameterized problem replaces, in polynomial time, the given instance by an equivalent instance whose size is a function of the parameter. The reduced instance is called a kernel for the problem. Once we have a kernel, a brute force algorithm on the resulting instance answers the problem in \fpt time. If the kernel size is a polynomial function of the input, then the problem is said to have a polynomial kernel.

\begin{definition}[Composition \cite{BDF}]
	A composition algorithm (also called OR-composition algorithm) for a parameterized problem $\Pi \subseteq \Sigma^{*} \times \mathbb{N}$ is an algorithm that receives as input a sequence $((x_1 , k), \cdots , (x_t , k))$,	with $(x_i , k) \in \Sigma^* \times \mathbb{N}$ for each $1 \leq i \leq t$, uses time polynomial in $\sum_{i=1}^t |x_i| + k$, and 	outputs $(y, k') \in \Sigma^* \times \mathbb{N}$ with  (a) $(y, k' ) \in \Pi \Longleftrightarrow (x_i , k) \in \Pi$ for some $1 \leq i \leq t$ 	and (b) $k'$ is polynomial in $k$. 
	A parameterized problem is compositional (or OR-compositional) if there is a composition algorithm for it.
	
\end{definition}

It is {\em unlikely} that an \npc problem has both a composition algorithm and a polynomial kernel as suggested by the following theorem.

\begin{theorem}[\cite{BDF,FS}]\label{composition_thm}
	Let $\Pi$ be a compositional parameterized problem whose unparameterized
	version $\tilde{\Pi}$ is NP-complete. Then, if $\Pi$ has a polynomial kernel then ${\sf co}$-${\sf NP}\subseteq {\sf NP/poly}$.
\end{theorem}
%A reduction rule that replaces an instance $(I,k)$ of a parameterized language $L$ by a reduced instance $(I',k')$ is said to be {\em safe}, if  $(I,k)\in L$ if and only if  $(I',k') \in L$. 

\noindent
{\bf Treewidth.} 
A tree decomposition of a graph $G$ is a pair $(T, \B)$ consisting of a tree $T$ and
a family $\B = \{B_t:t\in V(T)\}$ of sets $B_t \subseteq V(G)$, called bags, satisfying the following three conditions: 
\begin{description}
	\item[(I1)] $V(G)=\cup_{t\in V(T)}B_t$
	\item[(I2)] For every edge $(u,v)\in E(G)$, there exists a $t\in V(T)$ such that $u, v \in B_t$
	\item[(I3)] For every vertex $v\in V(G)$ the set, $T_v=\{t\in V(T)|v\in B_t\}$, induces a connected subtree of $T$. 
\end{description}

The width of a tree decomposition $(T, \B)$ is $\max{|B_t|- 1 : t \in V(T)}$. The
treewidth of $G$ is the minimum width over all tree decompositions of $G$. A tree decomposition $(T, \B)$ where $T$ is a rooted tree with root node $r$, is nice if  $B_r = \emptyset$ and $B_l = \emptyset$ for every leaf node $l$ of $T$. Further every non-leaf node of $T$ is of one of the following three types:

\begin{description}
	\item[Introduce vertex node:] a node $t$ with exactly one child $t'$ such that $B_t=B_{t'}\cup\{v\}$ for some $v\notin B_{t'}$. We say that $v$ is introduced at $t$.
	\item[Introduce edge node:] a node $t$, labeled with an edge $(u,v) \in E(G)$ such
	that $u, v \in B_t$, and with exactly one child $t'$ such that $B_t=B_{t'}$. We say	that edge $(u,v)$ is introduced at $t$.
	\item[Forget node:] a node $t$ with exactly one child $t'$ such that $B_t=B_{t'}\setminus\{v\}$ for some $v\in B_{t'}$. We say that $v$ is forgotten at $t$.
	\item[Join Node] a node $t$ with exactly two child $t_1$ and $t_2$ such that $B_t=B_{t_1}=B_{t_2}$.
\end{description}

% !TEX root = connected.tex

\section{Hardness of \vertexpart}

In the following theorem, we show that the above problem is \nph as well as \wh.
\begin{theorem}$^\star$\footnote{Theorems marked with $\star$ are true for a more constrained version of the \vertexpart, i.e $G[X]$ is not only connected but also contains a path on $k$ vertices.}
	The \vertexpart is \wh  parameterized by the solution size. 
\end{theorem}

\begin{proof}
	
	We show a reduction from the \cp. Let $\I (G,k)$ be an instance of the \cp. We create a new instance $\I' (G',k)$ of \vertexpart as follows.\\
	$V(G'))=V'= V\cup U \cup W$ where $V=V(G)$ induces a clique. For every non-neighbour $v_i,v_j\in V(G)$ we have a vertex $u_{ij}\in U$ where $v_i$ and $v_j$ are connected to $u_{ij}$ by an edge. For each vertex $u_{ij}\in U$ we have $k+1$ many vertices $\{w_{ij}^1, w_{ij}^2,\cdots, w_{ij}^{k+1}\}$ in $W$ where there is a edge between each $w_{ij}^{l}$ and  $u_{ij}$ where $1\leq l\leq k+1$ (see Figure \ref{fig:red_cmcs}). Formally $U=\{u_{ij}|(v_i,v_j)\notin E(G)\}$ and $W=\{w_{ij}^l|u_{ij}\in U \text{ and }1\leq l\leq k+1\}$. $E(G')=E_1\cup E_2\cup E_3$ where $E_1=\{(v_i,v_j)|(v_i,v_j)\in V\}$, $E_2=\{(v_a, u_{a,b})|u_{a,b}\in U\}\cup \{(v_b, u_{a,b})|u_{a,b}\in U\}$ and $E_3=\{(u_{ij},w_{ij}^l)|u_{ij}\in U, 1\leq l\leq k+1\}$.

	\begin{figure}[h!]
		\centering
		\includegraphics[width=.8\textwidth]{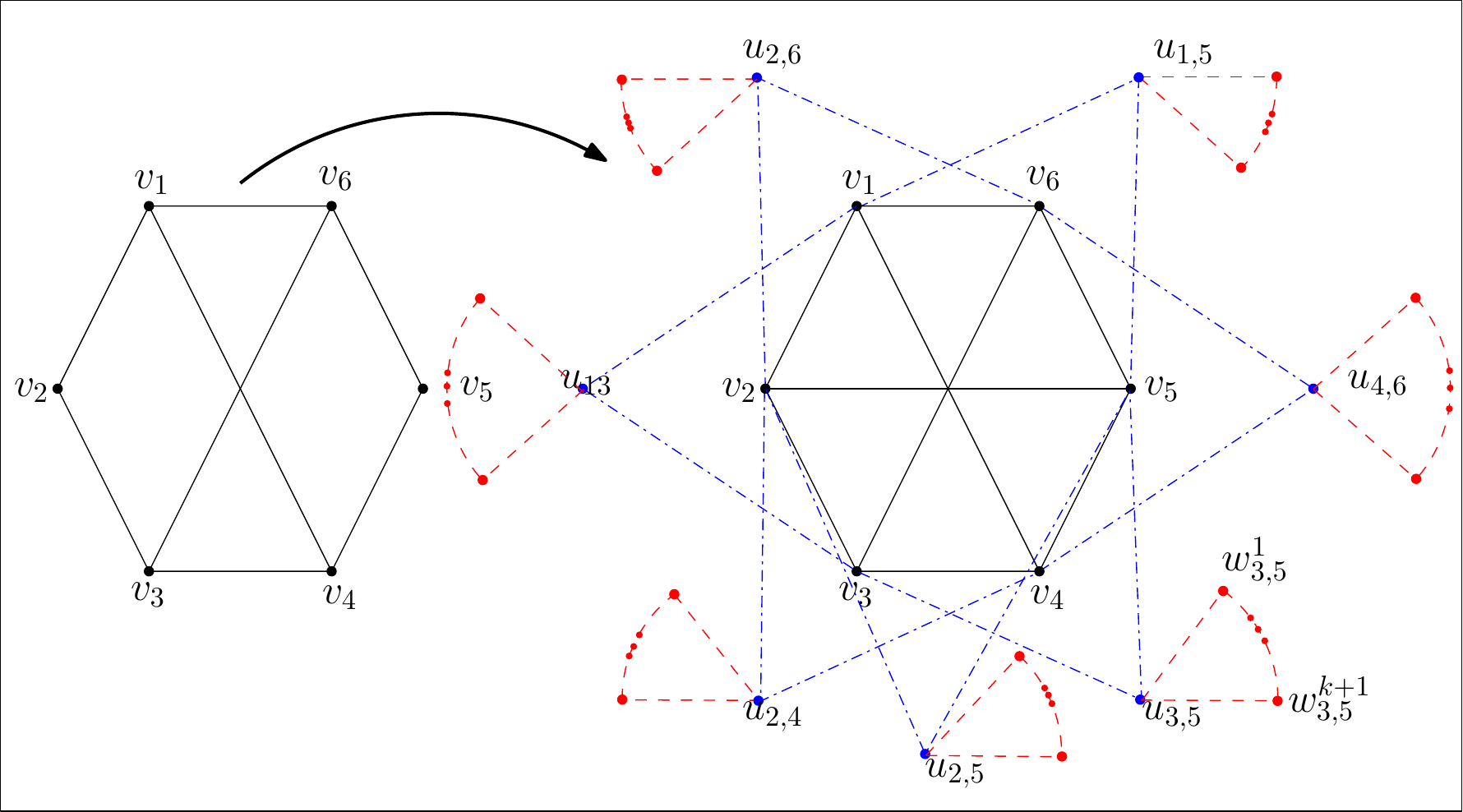}
		\caption{Construction of the \cmcs instance}
		\label{fig:red_cmcs} 
	\end{figure}
	\begin{claim}
		The $G$ contains a clique of size $k$ if and only if there exists a set $X \subseteq V(G')$ of cardinality $k$ such that $G'[V' \setminus X]$ and  $G'[X]$ both are connected.
	\end{claim}
	\begin{proof}
		
		Assume, $G$ contains a clique $S$ of size $k$. 
		Observe that $S$, also, induces a clique of size $k$ in the graph $G'$. Thus, $G[S]$ is connected. Now we show that $G'[V' \setminus S]$ is connected. 
		For the sake of contradiction, assume in $G'[V'\setminus S]$, there exist two vertices $x$ and $y$ which are not connected in $G'[V' \setminus S]$. 
		\begin{description}
			\item[Case 1: $x\in V\setminus S$ and $y\in V\setminus S$.] As $V\setminus S$ induces a clique in $G'$, there is an edge between $x$ and $y$. Therefore, contradiction.
			\item[Case 2: $x\in U$ and $y\in V\setminus S$ or vice versa.] Let $x=u_{i,j}$. As $v_i,v_j,y \in V$ they form a clique. Therefore, if either $v_i$ or $v_j$ is in $V\setminus S$ say $v_i$ then there is a path $(x,v_i,y)$ between $x$ and $y$ in $G'[V'\setminus S]$. On the other hand, $v_i$ and $v_j$ are nonadjacent in $G$ (by definition of $u_{ij}$). Thus, if $v_i,v_j\in S$ then it contradicts the fact that $S$ induces a clique in $G$.
			\item[Case 3: $x\in U$ and $y\in U$.] Let $x=u_{a,b}$ and $y=u_{c,d}$. Observe that as $S$ is a clique in $G$, and $v_a$ and $v_b$ are non-neighbours in $G$, $|S\cap \{v_a,v_b\}|\leq 1$. Without loss of generality assume $v_a\in V\setminus S$. Similarly, without loss of generality assume $v_c\in V\setminus S$. Now we have a path $(x,v_a,v_c,y)$ between $x$ and $y$ in $G[V'\setminus S]$. Thus, contradiction.
			\item [Case 4: At least one of $x$ or $y$ in $W$.] Let $x=w_{i,j}^l$. Observe that from the previous two cases (Case 2 and Case 3) we know that $u_{i,j}$ is connected to every vertex in $U$ and every vertex in $V\setminus S$. Hence, $u_{i,j}$ is connected to any vertex in $W$. Hence, $x$ is connected to all the vertices in $U\cup \{V\setminus S\}\cup W$. This concludes the proof.
		\end{description}

		Now we prove the reverse direction. Assume that in $G'$, there exists $X \subseteq V'$ of cardinality $k$ such that $G'[V' \setminus X]$ and $G'[X]$ both are connected. Observe that for all $u_{ij}\in U$, there are $k+1$ pendant vertices adjacent with $u_{ij}$. Thus, there exist at least one vertex $w_{ij}^l\notin X$. Hence, $u_{i,j}\notin X$ otherwise $w_{ij}^l$ will be isolated in $G[V'\setminus X]$. 
		
		Observe, as $X$ is connected, if there exist a vertex $w_{ij}^l\in X$ then $u_{ij}\in X$. We have ruled out that possibility, and hence $W\cap X=\emptyset$. 
		Hence, $X \subseteq V$. Next, we show that $X$ forms a clique in $G$. Assume that two non-neighbors $v_i,v_j\in X$.  Observe that $u_{ij}$ along with the vertices $\{w_{ij}^l|1\leq l\leq k+1\}$ is disconnected from the rest of the graph in $G[V'\setminus X]$. This contradicts the fact that $G[V'\setminus X]$ is connected. This concludes the proof.
	\end{proof}
	Observe that this reduction is a parameter preserving reduction where the parameters in both the cases are the same as $k$. Thus, this problem is \wh, parameterized by the solution size.
\end{proof}

% !TEX root = connected.tex

\section{\bcptwo on Graphs with Bounded Treewidth}\label{Sec:treewidth}
Treewidth is a measure to determine the tree-like structure of a graph $G$. Many computationally hard problems become polynomial time solvable for graphs with bounded treewidth. Therefore, treewidth is a natural parameter and many hard problems are shown to be \fpt parameterized by the treewidth of the underlying graph. In this section, we present a \fpt algorithm for \bcptwo parameterized by the treewidth of the given graph. More specifically, we will consider the following problem.
% \defparprob{\bcptwo}{An undirected graph $G$, a nice tree decomposition $(T, \B)$ of $G$ of width at most $k$ and an integer $\ell$}{k}{Does there exist a set of $\ell$ vertices $V_{\ell}\subseteq V(G)$ such that $G[V_{\ell}]$ and $G-V_{\ell}$ both are connected?}

\defparprob{\bf Balanced Connected $2$-Partition Problem (\bcptwo)}{An undirected graph $G$, two positive integers $n_1,n_2$ and  nice tree decomposition $(T, \B)$ of $G$ of width at most $k$ such that $n_1+n_2=|V(G)|$.}{k}{Does there exists a partition $\{V_1,V_2\}$ of $V(G)$ such that  $|V_i|=n_i$ and $G[V_i]$ is connected for all $1\leq i\leq 2$?}

Due to the famous Courcelle’s theorem~\cite{Courcelle}, it is also known that any problem expressible in {\bf MSO$_2$} is \fpt parameterized by treewidth.
\begin{theorem}\label{th:Courcelle}
	(Courcelle’s theorem, \cite{Courcelle}). Assume that $\varphi$ is a formula
	of {\bf MSO$_2$} and $G$ is a $n$-vertex graph equipped with evaluation of all the free variables of $\varphi$. Suppose, moreover, that a tree decomposition of $G$ of width $t$ is provided. Then there exists an algorithm that veriﬁes whether $\varphi$ is satisﬁed in $G$ in time $f (||\varphi||, t) \cdot n$, for some computable function $f$.
\end{theorem}

Observe that \bcptwo can be defined by {\bf MSO$_2$} formula as follows.
%\alert{write the MSO2 formula} 
$$\textsc{con-ver-per}(X)=\textsc{con}(X)\land \textsc{con}(V\setminus X)$$
where, for a graph $G$ on the set $V$ of vertices and $X\subseteq V$, it checks whether both $G[X]$ and $G[V\setminus X]$ are connected. We adopt the definition of $\textsc{con}(A)$ from~\cite{PCBook}.
\begin{multline*}
     \textsc{con}(A)=\forall_{B\subseteq V}[(\exists (a\in A ~\land ~a\in B) \land \exists (b\in A ~\land ~b\notin B)) \implies\\ (\exists_{e \in E}~\exists_{a\in A}~\exists_{b\in A} ~\textsc{inc}(a,e) \land \textsc{inc}(b,e) \land a\in B \land b\notin B)]
\end{multline*}
 where, given a graph $G=(V,E)$ and a subset $A$ of $V$, it checks whether $G[A]$ is connected.  $\textsc{inc}(u,e)$ checks whether the edge $e$ is incident on the vertex $u$. 

Therefore, by Theorem~\ref{th:Courcelle}, \bcptwo is \fpt.
Although Courcelle’s theorem provides an excellent tool to prove many problems to be \fpt parameterized by treewidth, in practice $f(||\varphi||, t)$ is a rapidly increasing function. Even for very simple {\bf MSO$_2$} formulas, this function can have super-exponential dependence on the treewidth of the graph \cite{10.5555/267871.267878}. Thus, efforts have been made to construct a dynamic programming-based algorithm with a single exponential dependency on the treewidth. In the rest of the section, we provide an efficient dynamic programming-based algorithm for \bcptwo with improved runtime.

Observe that if $n_1=|V(G)|$ or $n_1=0$, \bcptwo can be trivially solved. For ease of explanation from now onward, we assume that we are given two vertices $a$ and $b$ and the objective is to find a subset $V_1$ of size $n_1$ such that $a\in V_1,b\notin V_1$ and $G[V_1]$ and $G-V_1$ both are connected. 

We add $a$ and $b$ to every bag of $\B$.
Towards the construction of the dynamic programming, we define the following subproblem for each bag. 

For a node $t$ of $T$, let $V_t$ be the union of all the bags present in the subtree of $T$ rooted at $t$, including $B_t$. With each node $t$ of the tree decomposition, we associate a subgraph $G_t$ of $G$ defined as follows:
$$G_t =(V_t,E_t) \text{ where $E_t=\{e |~e$ is introduced in the subtree rooted at $t\}$}$$

We say the connected components $C_1, C_2, \cdots, C_b$ of a graph $G$ induce a partition $\Pe=\{P_1,P_2\cdots P_b\}$ of a vertex subset $X$ if $\forall i,~ V(C_i)\cap X= P_i $.
For a bag $B_t$ and an arbitrary subset of vertices  $V'\subseteq V(G)$, We define the following. $V_t^{V'}=V_t\cap V'$, $B_t^{V'}=B_t\cap V'$ and $G_t^{V'}=G_t[V']$. 

For any solution $(A,B)$, $G_t^A$ and $G_t^B$ will contain one or many connected components. Observe that each such connected component must intersect $B_t$, as $B_t$ is a separator between, $V_t\setminus B_t$ and the rest of the graph. We will capture this idea in the following recursive definition.

%Let $U$ and $W$ be any partition of $V_t$. 
%Observe that $G_t[U]$ and $G_t[W]$ is a collection of connected components. Let the partition of $B_t^U$ and $B_t^W$ induced by such connected components be $\Pe_t$ and $\Q_t$ respectively. 
For a bag $B_t$, a set of vertices $B_t^U$, a partition $\Pe=\{P_1,P_2,\cdots, P_p\}$ of $B_t^U$, a partition $\Q=\{Q_1,Q_2,\cdots, Q_q\}$ of $B_t^W=B_t\setminus B_t^U$ and an integer $1\leq n_u\leq n_1$,  $\eta(t,B_t^U,B_t^W,\Pe,\Q,n_u)$ is \true if there exists a partition $(U,W)$ of $V_t$ where

\begin{itemize}
	\item $U$ induces exactly $p=|\Pe|$ connected components in, $G_t[V_t]$ and those connected components induce the partition $\Pe=\{P_1,P_2,\cdots, P_p\}$ of $B_t^U=B_t\cap U$.
	\item  $W$ induces exactly $q=|\Q|$ connected components in, $V_t$ and those connected components induce the partition $\Q=\{Q_1,Q_2,\cdots, Q_q\}$ of $B_t^W=B_t\cap W$.
	\item $|V_t^U|=n_u$
\end{itemize}

% the part of $U$ contained in $G_t$ be $G_t^U$. Let $B_t^U=B_t\cap U$.
% Let $(U,W)$ be any solution.
% $G_t[S_t]$ and $G_t-S_t$ will contain one or many connected components. Observe that each such connected component must intersect $B_t$, as $B_t$ is a separator between $V_t\setminus B_t$ and the rest of the graph. Let the partition of $S_t$ and $R_t$ induced by such connected components be $\Pe_t$ and $\Q_t$ respectively. Let $p=|\Pe|$ and $q=|\Q|$.
% Consider the following function. For a bag $B_t$, a set of vertices $S_t$, a partition $\Pe=\{P_1,P_2\cdots P_p\}$ of $S_t$, a partition $\Q=\{Q_1,Q_2\cdots Q_q\}$ of $R_t=B_t\setminus P_t$ and an integer $1\leq \ell'\leq \ell$,  $\eta(t,S_t,R_t,\Pe,\Q,\ell')$ is \true if there exists a solution $S$ of where

% \textcolor{blue}{}

Please note that $B_t^W$ is a redundant parameter, it is kept in the definition for  ease of understanding. Note that there is a solution of size $(n_1,n_2)$ if $\eta(r,\{a\},\{b\},\{\{a\}\}, \{\{b\}\},n_1)$ is \true. Now we present a bottom up recursive computation of $\eta$ as follows

\vspace{10pt}
\noindent
{\bf Leaf Node.} If $t$ is a leaf node, set $\eta(t,\{a\},\{b\},\{\{a\}\}, \{\{b\}\},1)=\true$ 

\vspace{10pt}
\noindent
{\bf Introduce vertex node.} Let $v$ be the vertex introduced at $B_t$ and $t$ has only one child $t'$. As none of the edges incident on $v$ is introduced yet, $v$ must form a singleton set either at $\Pe$ or $\Q$. We present here the case when $\{v\}\in \Pe$. Set
$\eta(t,B_t^U,B_t^W,\Pe,\Q,n_u)=
\eta(t,B_t^U\setminus \{v\},B_t^W,\Pe\setminus \{\{v\}\},\Q,n_u-1)$ is \true,  set $\eta(t,B_t^U,B_t^W,\Pe,\Q,n_u)=\false$ otherwise.

% Set $\eta(t,S_t,R_t,\Pe,\Q,\ell')=\true$ if  $\eta(t',S_t\setminus \{v\} ,R_t,\Pe\setminus \{\{v\}\},\Q,\ell'-1) $ is \true,  set $\eta(t,S_t,R_t,\Pe,\Q,\ell')=\false$ otherwise.

\vspace{10pt}
\noindent
{\bf Introduce edge node.} Let the edge $(u,v)$ is introduced at $t$ and $t$ has only one child $t'$. There can be two cases. In the first case $u$ and $v$ either both belongs to $B_t^U$ or both belongs to $B_t^W$. In the second case one of the $u$ and $v$ belongs to $B_t^U$ and another in $B_t^W$.
We consider the Case 1 first. Without loss of generality assume $u,v\in B_t^U$.
% $u$ and $v$ must belong to the same partition of either $\Pe$ or $\Q$, otherwise we can set $\eta(t,S_t,R_t,\Pe,\Q,\ell')=\false$. We present the case when $u$ and $v$ both belongs to a same partition of $\Pe$. The case when $u$ and $v$ both belongs to a same partition of $\Q$, can be dealt similarly. WLG assume that $u,v\in P_1$. 

$$\eta(t,B_t^U,B_t^W,\Pe,\Q,n_u)=
\left[
\bigvee\limits_{\Pe'}  
\eta(t',B_t^U,B_t^W,\Pe\setminus P_1\cup \Pe',\Q,n_u)
\right]
\bigvee
\eta(t',B_t^U,B_t^W,\Pe,\Q,n_u)
$$

Next we consider the case when one of the $u$ and $v$ belongs to $B_t^U$ and another in $B_t^W$. In this case 
$$\eta(t,B_t^U,B_t^W,\Pe,\Q,n_u)=\eta(t',B_t^U,B_t^W,\Pe,\Q,n_u)$$
% $$\eta(t,S_t,R_t,\Pe,\Q,\ell')= \bigvee\limits_{\Pe'}    \eta(t',S_t,R_t,\Pe\setminus P_1\cup \Pe',\Q,\ell'+1) \bigvee  \eta(t',S_t,R_t,\Pe,\Q,\ell')$$
% Here $\Pe'$ is a partition of $P_1$ where $u$ and $v$ are in different partition.

\vspace{10pt}
\noindent
{\bf Forget Node.} Let $t$ be the forget node with $B_t=B_{t'}\setminus\{w\}$. 

\begin{equation*} 
\begin{split}
\eta(t,B_t^U,B_t^W,\Pe,\Q,n_u) & = \left[
\bigvee \limits_{\Pe'}
\eta(t',B_t^U\cup \{w\},B_t^W,\Pe',\Q,n_u)
\right] \\
&\text{\hspace{11pt}}
\bigvee\left[
\bigvee \limits_{\Q'}
\eta(t',B_t^U,B_t^W\cup \{w\},\Pe,\Q',n_u)
\right]
\end{split}
\end{equation*}

% $$\eta(t,B_t^U,B_t^W,\Pe,\Q,n_u)=
% \left[
% \bigvee \limits_{\Pe'}
% \eta(t',B_t^U\cup \{w\},B_t^W,\Pe',\Q,n_u)
% \right]
% \lor
% \left[
% \bigvee \limits_{\Q'}
% \eta(t',B_t^U,B_t^W\cup \{w\},\Pe,\Q',n_u)
% \right]
% $$
Here $\Pe'$ is a partition of $S_t\cup \{w\}$ { where, the vertex $w$ is being added to one of the parts in $\Pe$} and $\Q'$ is a partition of $R_t\cup \{w\}$ { where, the vertex $w$ is being added to one of the parts in $\Q$}.

%$$\eta(t,S_t,R_t,\Pe,\Q,\ell')= \bigvee \limits_{\Pe'}    \eta(t',S_t\cup \{w\},R_t,\Pe',\Q,|\Pe'|\textcolor{blue}{~(or~\ell'?)}) \bigvee\limits_{\Q'}  \eta(t',S_t,R_t\cup \{w\},\Pe,\Q',\ell')$$

\vspace{10pt}
\noindent
{\bf Join Node.} Suppose $t$ is a join node with children $t_1$ and $t_2$ . Recall that
$B_t=B_{t_1}=B_{t_2}$. Our objective here is to merge two partial solutions: one originating from $G_{t_1}$ and the second originating from $G_{t_2}$. Towards that, we define the following. For any set of vertices $A$ and a partition of $A$, $\A=\{A_1,A_2,\cdots, A_a\}$ we define a graph $G_{\A}$ with vertex set $A$ as follows. There is an edge between two vertices $a_1$ and $a_2$ if and only if $a_1$ and $a_2$ belongs to same partition of $\A$. For any set of vertices, $A$ we say two partitions of $A$, $\A_1$ and $\A_2$ \represents a third partition $\A$ if $G_{\A}=G_{\A_1}\cup G_{\A_1}$. 

$$\eta(t,B_t^U,B_t^W,\Pe,\Q,n_u)= \bigvee \limits_{\Pe_1,\Pe_2,\Q_1,\Q_2,n_u^1} 
[  
\eta(t_1,B_t^U,B_t^W,\Pe_1,\Q_1,n_u^1) 
\land 
\eta(t_2,B_t^U,B_t^W,\Pe_2,\Q_2,n_u^2)  
]$$

The OR runs over all possible choice of $\Pe_1,\Pe_2,\Q_1,\Q_2,n_u^1$ where $\Pe_1,\Pe_2$ \represents $\Pe$, $\Q_1,\Q_2$ \represents $Q$, $0\leq n_u^1\leq n_u$ and $n_u^2=n_u-n_u^1$. This concludes the dynamic programming description.

\vspace{10pt}
\noindent
{\bf Runtime.} The number of subproblems is determined by the number of different partition of $S_t$ and $R_t$ which is essentially bounded by the number of partition of $B_t$. The size of $B_t$ is bounded by $k_2$ and thus the number of partitions is bounded by $k^{\oh(k)}$. Time required to solve each sub-problem is  $k^{\oh(k)}$. Hence, the total runtime is $k^{\oh(k)}$.

\vspace{10pt}
\noindent
{\bf Note.} Runtime of the algorithm can be improved to $\oh(c^k)$ for some constant $c$ using standard techniques using representative sets.

%\subsection{Dynamic Programming Solution using Representative Sets}
%In this section we present a dynamic progamming solution with improved runtime. To improve over the $k^{\oh(k)}$ solution we employ the idea of representative sets defined over a graphic matroid which we define as follows. For a bag $t$ consider an arbitrary partition  $B_t^W$ and $B_t^U$ of $B_t$. 
%
%Let $\M^W_t=(E^W_t,\I^W_t)$ (respectively $\M^W_t=(E^U_t,\I^U_t)$) be a graphical matroid defined on the complete graph $K^W$ on $B_t^W$. Here the element set is the edge set of $K^W$ and the independent set $\I^W_t$ consists of the spanning forrests of $K^W$. We define  $\M^W_t$ similarly. For any partition  $B_t^W$ and $B_t^U$ of $B_t$ we define a set of forrests $\F_t^W\subset \I^W_t$ as follows.
%
%$$\F_t^W = \{F:\text{ for any edge }(u,v)\in E(F) \text{ there exist a path between $u$ and $v$ in $G_t$}\}\}$$
%
%We define $\F_t^W$ similarly.

% !TEX root = connected.tex

\section{\vertexpart on Planar Graphs}\label{sec:planar}
\subsection{\fpt Algorithm for \vertexpart on Planar Graphs}\label{sec:planarfpt}
In this section, we give an \fpt algorithm for \vertexpart. Let us first consider a restricted version of the \vertexpart, where we are given a vertex $v$ such that $v\in X$. The problem definition is as follows. 

\defprob{\bf \rvertexpart}{An undirected graph $G(V,E)$ with $n$ vertices, an integer $k$ and a vertex $v$.}{Does there exist $X \subseteq V$ of cardinality  $k$ such that $G[X]$  and $G[V \setminus X]$ is connected and $v\in X$?}
Next, we show a relation between the original version and the restricted version of the problem in the following lemma.

\begin{lemma}         
	If \rvertexpart can be solved in \fpt time then the \vertexpart can also be solved in time \fpt.
\end{lemma}
\begin{proof}
	Assume that the \rvertexpart can be solved in time \fpt. Observe that there are $n$ many choices for $v$, where $n$ is the number of vertices in the given graph. Thus given a graph we can create $n$ many restricted problem instances for this problem. Hence \vertexpart also can be solved in \fpt time.
\end{proof}

Let us first focus on the restricted version of the problem. Let the given vertex be $v$.  Define $V_k$ as the set of the vertices we reach till $k$-th level if we start doing breath first search (BFS) from the the vertex $v$. Since $|X|=k$ and $v_i\in X$, $X\subseteq V_k$. Observe that the graph $G[V\setminus V_k]$ might have several connected components say,  $G_1, G_2,\cdots, G_\ell$. Contract each connected component $G_i$ into a single vertex $u_i$. Let $U=\{u_i|i \leq \ell\}$. Next we consider the graph $G'=G[V_k\cup U]$ and make $N(u_i)= N(V(G_i))$ in $G'$. Observe that the diameter of $G'$ is $\oh(k)$. It is well known fact that a planar graph with diameter $d$ has treewidth $\oh(d)$. Thus treewidth of $G'$ is also bounded by $\oh(k)$. 

Now, using the algorithm described in Section~\ref{Sec:treewidth} we get the following theorem.
\begin{theorem}
    The \vertexpart on planar graphs can be solved in time $k^{\oh(k)}$.
\end{theorem}

% all the points in $X$ will be inside the circle $\mathcal{C}$ drawn centering $v_i$ with radius $k$. Let $W$ be the set of vertices inside the circle $\mathcal{C}$. Now, consider the graph $G[V\setminus W]$. The graph might have divided into several connected components $G_1, G_2,\cdots, G_\ell$. Contract each connected component $G_i$ into a single vertex $u_i$ and $N(u_i)= N(V(G_i))$ in $G$. Let $U=\{u_i|i \leq \ell\}$. Next we consider the graph $G'=G[V\cup U \setminus W]$. Observe that the diameter of $G'$ is $\oh(k)$. It is well known fact that a planar graph with diameter $d$ has treewidth $\oh(d)$. Thus treewidth of $G'$ is also bounded by $\oh(k)$. 

\subsection{No Polynomial Kernels for \vertexpart in Planar Graph}\label{sec:planarnopolykernel}
In this section, we prove that no polynomial sized kernel exists for \vertexpart on planar graphs. %Towards that we first explain one of the tools to prove such a lower bound, called composition.
Towards that, we show a composition for \cmcs on planar graphs in the following theorem.

\begin{theorem}$^\star$
	\vertexpart on planar graph parameterized by the solution size does not have a polynomial kernel.
\end{theorem}
\begin{proof}
	Let we are given $t$ many planar graphs $G_1,G_2,\cdots,G_t$ where each $G_i$ contains $n_i$ many vertices. We construct a new graph $G'$ as follows. We create $k+1$ many copies for each planar graph $G_i$. Since each planar graph has planar embedding keeping any vertex in the outer face, let us assume $G_i^1,G_i^2,\cdots,G_i^{k+1}$ be the embeddings of $G_i$ keeping $v_i^1,v_i^2,\cdots,v_i^{k+1}$ on the outer face, respectively. Now add an edge between $v_i^{\ell}$ and $v_i^{\ell+1}$ for all $1 \leq \ell \leq k$. Also, add an edge between $v_i^{k+1}$ and $v_{i+1}^{1}$ for all $1 \leq i \leq t-1$.
	
	\begin{claim}
		$G'$ is a $ \sf YES$ instance if and only if at least one of the $G_i$ is a $\sf YES$ instance.
	\end{claim}
	
	\begin{proof}
		We prove the forward direction first. Assume that in $G'$, there exists $X \subseteq V(G')$ of cardinality $k$ such that $G[X]$ and $G[V(G') \setminus X]$ both are connected. Since, deleting any vertex $v_i^j$ from $G_i^j$ where $j\in[k+1]$ makes the graph $G'$ disconnected. Thus $X$ can not contain any vertex $v_i^j$ from $G_i^j$ where $j \in [k+1]$. Hence, we can assume that $X\cap \{v_i^j\}=\emptyset$ in $G_j$ for $1\leq j \leq k+1$, i.e $X$ is contained completely inside any one of $G_i^j$ without containing any vertex at the outer face that shares edges with other copies of $G_i$. Therefore, $G[X]$ is a connected graph on $k$ vertices inside $G_i^j$ and $G[V(G_i^j)\setminus X]$ is connected. Hence $G_i$ is a $\sf YES$ instance.
		
		Now, we prove the reverse direction. Assume that there is a graph, $G_i$ in which there exists $X \subseteq V(G_i)$ of cardinality $k$ such that $G[X]$ and $G[V(G_i) \setminus X]$ both are connected. Without loss of generality, assume $X = \{v_1,v_2,\cdots,v_k\}$. Since $G_i$ contains $X$, all the copies of $G_i$ in $G'$ also contain $X$. Now, we show that there exists a copy of $G_i$, say $G_i^p$ in $G'$, that does not contain any vertex of $X$ on the outer face which shares edges with other copies of $G_i$. As $|X|=k$, observe that there can be at most $k$ many different copies of $G_i$ in $G'$ which has a vertex of $X$ on the outer faces which shares edges with the other copies. Hence by pigeon hole principle, there exists at least a copy of $G_i$ in $G'$ which has no vertex on the outer face that shares edges with other copies. Thus, deleting $X$ from that copy will not disconnect any $G_i^a$ or $G_j^b$ in the graph $G'$. Therefore, $G'$ is also a $\sf YES$ instance.
	\end{proof}
\end{proof}
 
% !TEX root = connected.tex

\section{\vertexpart on Unit Disk Graphs} \label{sec:udg}
Let we are given an unit disk graph $G(V,E)$ and its representation $(\D,C)$. Let $V=\{v_1,v_2, \cdots, v_n \}$ and for each $v_i$, $D_i$ be the corresponding disk centered at $c_i$. Here $\D=\{D_i|~1\leq i\leq n\}$ and $C=\{c_i|~1\leq i \leq n\}$. 
\begin{red_rule}\label{red_rule:udg_components}
	If the graph $G$ has more than two connected components the return $\sf NO$.
\end{red_rule}
\begin{lemma}
	\rr~\ref{red_rule:udg_components} is safe and can be implemented in polynomial time. 
\end{lemma}
\begin{proof}
	Assume $G$ has at least three connected components. After deleting any subset $X$ of cardinality $k$ from the graph such that $G[X]$ is connected, the graph remains disconnected. The number of components in the given graph can be checked in polynomial time. \end{proof}

Let us assume that the graph has two connected components. Observe that the given instance is a \yes instance if and only if at least one of the two components contains exactly $k$ vertices. Otherwise, the instance is a \no instance. The number of components in the given graph and the number of vertices in each component can be checked in polynomial time. Thus now onward we assume that $G$ is connected. Let us consider a $(\frac{1}{2}\times \frac{1}{2})$ square grid on the plane. Assume $V_S$ be the set of centers of the disks that are contained inside the cell $S$; more formally, $V_S=\{v_i|~c_i\in S\}$. We also define $N(V',S)$ as the set of vertices in the cell $S$ that are neighbors of vertices in $V'$. For any cell $S$ in the grid, next, we prove that if $S$ contains at least $k+24$ centers then the given instance is a \yes instance.

\begin{red_rule}\label{red_rule:udg_k+24}
	If there exists a cell $S$ containing at least $k+24$ many centers then return \yes.
\end{red_rule}
\begin{lemma}\label{lemma:udg_k+24}
	\rr~\ref{red_rule:udg_k+24} is safe and can be implemented in polynomial time.
\end{lemma}
\begin{proof}
	Select an arbitrary cell $S$ with at least $k+24$ centers. Observe that any vertex $v\in V_S$ can have a neighbor inside the circle drawn centering $v$ with radius one. That circle intersects at most $25$ cells including $S$. Observe that, these $25$ cells form a $5\times 5$ grid keeping $S$ as the center.
	Let $A_S$ be the set of cells where a vertex $v\in V_S$ can have neighbors. Now, for each such cell $S'\in A_S$ if $N(V_{S'}, S)\neq \emptyset$, then mark an arbitrary vertex in $N(V_{S'}, S)$. As $v$ can have neighbors in other $24$ cells, we can mark at most $24$ vertices in $S$. Thus there are at least $k$ unmarked vertices in $S$. Let $X$ be a set of any arbitrary $k$ many unmarked vertices. Observe that, any two disks having a center inside a cell intersect. Thus, all the centers in a cell form a clique. Hence $G[X]$ is also a clique on $k$ vertices. Observe that to show $G[V\setminus X]$ is connected it is enough to show that all the vertices in $N[X]$ are connected in $G[V\setminus X]$. Let $x$ and $y$ be any two vertices from $N[X]$ in $G[V\setminus X]$. We show that there exists a path between $x$ and $y$ in $G[V\setminus X]$. There can be three cases based on where the vertices $x$ and $y$ lie. First, both $x, y$ in $S$. As all the vertices in $S$ form a clique, after deletion of $X$ remaining vertices in $S$ also induce a clique. Thus $x$ and $y$ are adjacent. The second case is when exactly one of $x$ and $y$ is in $S$. Let $x\in S$ and $y\in S'$. If any vertex $q$ in $S'$ was adjacent to at least one vertex $p$ in $S$ then there is marked vertex $p$ in $S$ which is not in $X$. Observe that, $p$ is adjacent to $x$ and $q$ is adjacent to $y$. Hence there exists a path between $x$ and $y$ in $G[V\setminus X]$. If $S'$ does not have any vertex adjacent to a vertex in $S$, then $y$ is connected to $x$ via vertices in other cell $S''$. Similarly, we can prove that any vertex in $S''$ is connected to $x$. Hence $x$ and $y$ are connected by a path. Now we consider the last case where both $x$ and $y$ are not in $S$. Assume $y\in V_{S'}$ and $x\in V_{S''}$. If any vertex $q$ in $S'$ was adjacent to at least one vertex $p$ in $S$ then there is marked vertex $p$ in $S$ which is not in $X$. Observe that, $q$ is adjacent to $y$. Again, if there exists a vertex $r$ in $S''$ that is adjacent to at least a vertex $s$ in $S$ then we can argue similarly and conclude that there exists a path between $x$ and $s$ in $G[V\setminus X]$. Observe that, $r$ and $p$ are adjacent in  $G[V\setminus X]$, which connects $x$ and $y$ in $G[V\setminus X]$. Moreover, we can find such a subset $X$ of vertices in polynomial time. This concludes the proof.
	%in similar way we can prove that there exist a vertex $p$ in $S$  which is connected to $x$ and a vertex $q$ which is connected to $y$ in $G[V\setminus X]$. Also observe, $p$ and $q$ are adjacent here. Thus there is a path between $x$ and $y$ in $G[V\setminus X]$.
\end{proof}

Now onward we can assume that each cell has at most $k+23$ centers. Next, we give a turing kernel for \vertexpart on unit disk graphs in the following theorem. Toward doing so we first prove the following theorem.

\begin{theorem}$^\star$\label{th:udg_kernel}
	The \rvertexpart has a $\oh(k^3)$ kernel on unit disk graphs.
\end{theorem}
\begin{proof}
	%Let $S$ be any arbitrary cell in the grid and $v$ be any vertex in $(V_S\cap X)$. 
	Let $v\in X$ be the given vertex. Observe that all the vertices in $X$ lie inside the circle drawn centering $v$ with radius $k$. Thus there are at most $\oh (k^2)$ many cells which can contain the other vertices of $X$. By \rr~\ref{red_rule:udg_k+24}, no cell contains more than $k+23$ vertices. Thus there can be at most $\oh(k^3)$ vertices, which can be the other endpoint of the path. Therefore we have $\oh(k^3)$ kernel.
	%Let $v\in X$ be a vertex  from $S$. Assume that, $v$ be one endpoint of the path contained in $G[X]$. Observe that, the other endpoint can lie inside the circle drawn centering $v$ with radius $k$. Thus there are at most $\oh (k^2)$ many cells which can contain the other endpoints. By \rr~\ref{red_rule:udg_k+24}, no cell contains more than $k+23$ vertices. Thus there can be at most $\oh(k^3)$ vertices, which can be the other endpoint of the path. Therefore we have $\oh(k^3)$ kernel. 
\end{proof}
Observe that from the graph $G$, for each $v_i$ we can create a restricted problem. From Theorem~\ref{th:udg_kernel}, \rvertexpart has a $\oh(k^3)$ kernel. Thus we have the following corollary.
\begin{corollary}$^\star$
	The \vertexpart on unit disk graphs admits a polynomial turing kernel.
\end{corollary}
% !TEX root = connected.tex

\section{\fpt Algorithm for \edgepart}\label{sec:edge}
In this section, we present a \fpt algorithm to solve the \edgepart. The idea used here is similar to the \fpt algorithm for the  Longest path problem~\cite{}. Before we proceed with the algorithm, let us introduce few well-known results which will be essential for our algorithm.

Matroid theory has been used widely to develop many \fpt algorithms. In this section, we use representative sets defined on co-graphic matroid to design \fpt algorithm for \edgepart. We present a brief introduction to matroid and representative sets for the sake of completeness. For a more detailed discussion on these topics please refer to~\cite{PCBook,MatroidBookOxlay}.
\begin{definition}[\cite{MatroidBookOxlay}]
	A matroid $M$ is an ordered pair $(E, \I)$ consisting of a ﬁnite set $E$ and a
	collection $\I$ of subsets of $E$ having the following three properties:
	\begin{description}
		\item[(I1)] $\emptyset \in \I$.
		\item[(I2)] If $I\in \I$ and $I'\subset I$, then $I'\in\I$.
		\item[(I3)] If $I_1,I_2\in\I$ and $|I_1|<|I_2|$, then there is an element $e\in I_1\setminus I_1$ such that $I_1\cup\{e\}\in\I$.
	\end{description}
\end{definition}	

Let $M(E, \I)$ be a matroid. The members of $\I$ are called the independent sets of $M$, and $E$ is the ground set of $M$. For any two sets $A,B\subset E$, we say that $A$ ﬁts with $B$ if $A$ is disjoint from $B$ and
$A \cup B$ is independent.

\begin{definition}[\cite{PCBook}]
	Let $M$ be a matroid and $\A$ be a family of sets of size $p$ in $M$. A subfamily $\A'\subseteq \A$ is said to $q$-represent $\A$ if for every set $B$ of size
	$q$ such that there is an $A \in \A$ that ﬁts $B$, there is an $A'\in \A'$ that also ﬁts
	$B$.  If $\A'$ $q$-represents $\A$, we write $\A'\repset{q} \A$.
\end{definition}

Next, we have the following theorem where $\omega$ is the matrix multiplication constant.
\begin{theorem}[\cite{PCBook}]\label{theorem:repset}
	There is an algorithm that, given a matrix $M$ over a ﬁeld $GF(s)$, representing a matroid $M = (U, F)$ of rank $k$, a $p$-family $\A$ of independent sets in $M$, and an integer $q$ such that $p  + q  = k$, computes a $q$-representative family $\A'\repset{q}\A$ of size at most ${{p+q} \choose {p}}$ using at most $\oh(|\A|({{p+q} \choose {p}} p^\omega+{{p+q} \choose {p}}^{\omega-1}))$
	operations over $GF(s)$.
\end{theorem}

Let $G$ be a graph. Let us define $\cgm(G)=(E,\I)$ where $\I=\{E\subseteq E(G)| G-E \text{ is connected}\}$. 
$\cgm(G)$ is called the co-graphic matroid.
For any vertex $v$, let $\E_v^p$ be a set of subset of $p$ edges $X$ such that $G[X]$ is connected, $G[X]$ contains $v$ and $G-X$ is connected. Formally
$$ \E_v^p=\{ X|~X\subseteq E(G), |X|=p, \text{ $G[X]$ and $G-X$ is connected}    \}$$

Observe that for all $v\in V$ and $1\leq p\leq k$, $\E_{v}^p$ is a $p-$family independent set in $\cgm(G)$. We can define $q-$repesentative family for $\E_{uv}^p$ with respect to $\cgm(G)$.

$$\F_{v}^{p,q}\repset{q} \E_{v}^p$$

Next we present a dynamic programming based algorithm to compute $\F_{v}^{p,q}$ for all $v\in V(G)$, $1\leq p\leq k$ and $1\leq q\leq k-p$.

\begin{theorem}
	There is an algorithm which computes for all $v\in V(G)$, $1\leq p\leq k$ $1\leq q\leq k-p$, $\F_{v}^{p,q}$ in time $\oh^*(5.2^k)$.
\end{theorem} 

\begin{proof}
	Observe that $\E_{v}^1=\{(u,v) :(u,v)\in E(G)\}$. We set $\F_{v}^{1,k-1}=\E_{uv}^1$. We compute $\F_{v}^{p,q}$ in a bottom up fashion using the following recursive definition.
	We create an auxiliary set $\mytilde{\E_{v}^{p}}$ as follows
%	$$\mytilde{\F_{v}^{p,q}}=\{ 
%	\tilde{E}|\text{ where }\tilde{E}= E'\cup (u,v) \text{ } \forall \text{ } (u,v)\in E(G),\forall X\in \F_{u}^{p-1,q+1} \text{ and } G-\tilde{E} \text{ is connected }
%	      \}$$
	      
		$$\mytilde{\E_{v}^{p}}=\{ 
	E'\cup (u,v)|\text{ where } (u,v)\in E(G), X\in \F_{u}^{p-1,q+1} \text{ and } G-\{E'\cup (u,v)\} \text{ is connected }
	\}$$

	Observe that cardinality of  $\mytilde{\E_{v}^{p}}$ is at most ${{p+q}\choose {p}}|E(G)|$ and $\mytilde{\E_{v}^{p}}$ is also a $p-$family independent set in $\cgm(G)$. Thus we can define a $q-$representative family $\mytilde{\F_{v}^{p,q}} \repset{q} \mytilde{\E_{v}^{p}}$.
	We can compute $\mytilde{\F_{v}^{p,q}}$  using the algorithm of Theorem~\ref{theorem:repset} in time $|\E_{uw}^{p-1,q+1}| {\choosepq^{\omega-1}} n^{\oh(1)}$ which is equals to ${\choosepq^{\omega}} n^{\oh(1)}$. Note that as $q-$representation is transitive~\cite{PCBook}, $\mytilde{\F_{v}^{p,q}}$ will $q-$represent $\E_{v}^p$. We set $\F_{v}^{p,q}=\mytilde{\F_{v}^{p,q}}$. Overall runtime of the algorithm is $\sum_{p=2}^{k} {\choosepq^{\omega}} n^{\oh(1)}$ which is $\oh^*(2^{k\omega})$. Due to a recent result~\cite{mm21}, it is known that $\omega\leq2.37286$. Therefore the runtime of the algorithm is bounded by $\oh^*(5.2^k)$.
	
\end{proof}

\bibliographystyle{plain}
\bibliography{reference}

\begin{thebibliography}{10}

\bibitem{mm21}
Josh Alman and Virginia~Vassilevska Williams.
\newblock A refined laser method and faster matrix multiplication.
\newblock In D{\'{a}}niel Marx, editor, {\em Proceedings of the 2021 {ACM-SIAM}
  Symposium on Discrete Algorithms, {SODA} 2021, Virtual Conference, January 10
  - 13, 2021}, pages 522--539. {SIAM}, 2021.

\bibitem{BDF}
Hans~L. Bodlaender, Rodney~G. Downey, Michael~R. Fellows, and Danny Hermelin.
\newblock On problems without polynomial kernels.
\newblock {\em J. Comput. Syst. Sci.}, 75(8):423--434, 2009.

\bibitem{camerini1983complexity}
Paolo~M Camerini, Giulia Galbiati, and Francesco Maffioli.
\newblock On the complexity of finding multi-constrained spanning trees.
\newblock {\em Discrete Applied Mathematics}, 5(1):39--50, 1983.

\bibitem{chataigner2007approximation}
Fr{\'e}d{\'e}ric Chataigner, Liliane~RB Salgado, and Yoshiko Wakabayashi.
\newblock Approximation and inapproximability results on balanced connected
  partitions of graphs.
\newblock {\em Discrete Mathematics and Theoretical Computer Science},
  9(1):177--192, 2007.

\bibitem{chlebikova1996approximating}
Janka Chleb{\'\i}kov{\'a}.
\newblock Approximating the maximally balanced connected partition problem in
  graphs.
\newblock {\em Information Processing Letters}, 60(5):225--230, 1996.

\bibitem{Courcelle}
Bruno Courcelle.
\newblock The monadic second-order logic of graphs. i. recognizable sets of
  finite graphs.
\newblock {\em Inf. Comput.}, 85(1):12–75, March 1990.

\bibitem{PCBook}
Marek Cygan, Fedor~V. Fomin, Lukasz Kowalik, Daniel Lokshtanov, Daniel Marx,
  Marcin Pilipczuk, Michal Pilipczuk, and Saket Saurabh.
\newblock {\em Parameterized Algorithms}.
\newblock Springer Publishing Company, Incorporated, 1st edition, 2015.

\bibitem{diestel2012graph}
Reinhard Diestel.
\newblock Graph theory, volume 173 of.
\newblock {\em Graduate texts in mathematics}, page~7, 2012.

\bibitem{FS}
Lance Fortnow and Rahul Santhanam.
\newblock Infeasibility of instance compression and succinct {PCP}s for {NP}.
\newblock In {\em STOC}, pages 133--142, 2008.

\bibitem{frank1975combinatorial}
A~Frank.
\newblock Combinatorial algorithms, algorithmic proofs.
\newblock {\em Doctoral Thesis, Eotvos University}, 1975.

\bibitem{gyori1976division}
Ervin Gyori.
\newblock On division of graphs to connected subgraphs, combinatorics.
\newblock In {\em Colloq. Math. Soc. Janos Bolyai, 1976}, 1976.

\bibitem{lovasz1977homology}
L{\'a}szl{\'o} Lov{\'a}sz.
\newblock A homology theory for spanning tress of a graph.
\newblock {\em Acta Mathematica Hungarica}, 30(3-4):241--251, 1977.

\bibitem{MatroidBookOxlay}
James~G. Oxley.
\newblock {\em Matroid Theory (Oxford Graduate Texts in Mathematics)}.
\newblock Oxford University Press, Inc., USA, 2006.

\bibitem{10.5555/267871.267878}
Wolfgang Thomas.
\newblock {\em Languages, Automata, and Logic}, page 389–455.
\newblock Springer-Verlag, Berlin, Heidelberg, 1997.

\end{thebibliography}
\end{document}